\documentclass[conference]{IEEEtran}

\usepackage{amsmath}
\usepackage{amsfonts}
\usepackage{amssymb}
\usepackage{amsbsy}
\usepackage{graphicx}
\usepackage{times}
\usepackage{default}
\usepackage{color}
\usepackage[noadjust]{cite}
\usepackage{enumerate}
\usepackage{tikz}
\usepackage{subcaption}

\providecommand{\lk}{\langle}
\providecommand{\rk}{\rangle}

\newcommand{\snr}{\mathsf{snr}}
\newcommand{\dof}{\mathsf{dof}}
\newcommand{\Dof}{\mathsf{DoF}}
\newcommand{\udim}{\overline{d}}
\newcommand{\ldim}{\underline{d}}

\renewcommand{\tilde}{\widetilde}

\def\ba#1\ea{\begin{align*}#1\end{align*}}	
\def\ban#1\ean{\begin{align}#1\end{align}}	
\def\bac#1\eac{\vspace{\abovedisplayskip}{\par\centering$\begin{aligned}#1\end{aligned}$\par}\addvspace{\belowdisplayskip}}	

\newcommand{\lefto}{\mathopen{}\left}

\newtheorem{theorem}{Theorem}

\newtheorem{definition}{Definition}

\newtheorem{remark}{Remark}

\IEEEoverridecommandlockouts 

\title{Explicit and Almost Sure Conditions for\\ $K/2$ Degrees of Freedom}

\author{
\IEEEauthorblockN{David Stotz and  Helmut B\"olcskei}\vspace{.3cm}
\IEEEauthorblockA{
       Dept.~IT \& EE, ETH Zurich, Switzerland\\
         {Email: \{dstotz,\ boelcskei\}@nari.ee.ethz.ch} 
\thanks{The authors would like to thank M.~Einsiedler for helpful discussions and for drawing their attention to \cite{Hoc12}.}
}}

\begin{document}
\maketitle
\begin{abstract}
 It is well known that in $K$-user constant single-antenna interference channels $K/2$ degrees of freedom (DoF) can be achieved for almost all channel matrices. Explicit conditions on the channel matrix to admit $K/2$ DoF are, however, not available. The purpose of this paper is to  identify such explicit conditions, which  are satisfied for almost all channel matrices.  We also provide a construction of  corresponding asymptotically DoF-optimal input distributions. The main technical tool used is a recent breakthrough result by Hochman in fractal geometry \cite{Hoc12}.  

\end{abstract}

\section{Introduction}
Characterizing the degrees of freedom (DoF) in interference channels (ICs) under various assumptions on the channel matrix has become a heavily researched topic in recent years \cite{Jaf11}. 
A particularly surprising result states that $K/2$ DoF can be achieved in single-antenna $K$-user ICs with constant channel matrix \cite{MGMK09}. This statement was shown to hold for (Lebesgue) almost all channel matrices \cite[Thm.~1]{MGMK09}. The  technical arguments---from Diophantine approximation theory---used in the proof of \cite[Thm.~1]{MGMK09} do not seem to allow an explicit characterization of the set of these channel matrices.
 What is known, though, is that channel matrices with all entries rational admit strictly less than $K/2$ DoF \cite{EO09} and hence belong to the set of exceptions relative to the ``almost-all result'' in \cite{MGMK09}.


Recently, Wu et al. \cite{WSV13} developed a general framework, based on (R\'enyi's) information dimension, for characterizing the DoF in constant single-antenna ICs. While this general and elegant theory allows to recover, inter alia, the ``almost-all result'' from  \cite{MGMK09}, it does not provide insights into the structure of the set of channel matrices admitting $K/2$ DoF.




\paragraph*{Contributions}
The main contribution of this paper is to complement the results in \cite{MGMK09, EO09, WSV13} by providing \emph{explicit and almost surely satisfied} conditions on the channel matrix to admit $K/2$ DoF. The conditions we find essentially require that the set of all monomial expressions in the channel coefficients  be linearly independent over the rational numbers. The proof of our main theorem employs a recent breakthrough result from fractal geometry  \cite{Hoc12}, which allows us to compute the information dimension of self-similar distributions under much milder conditions than the previously required open set condition \cite{BHR05}. For the channels satisfying our explicit and almost sure conditions, we furthermore provide a construction of asymptotically DoF-optimal input distributions.


Finally, we show that the explicit sufficient conditions for  $K/2$ DoF we identify are not necessary. This is accomplished by constructing examples of channel matrices that admit $K/2$ DoF but do not satisfy our conditions. The set of all such channel matrices, however, necessarily has Lebesgue measure zero.  



\paragraph*{Notation}
Random variables are represented by uppercase letters from the end of the alphabet. Lowercase letters are used  exclusively for deterministic quantities. Boldface uppercase letters  indicate matrices.  Sets are denoted by uppercase calligraphic letters. For $x\in\mathbb R$, we write $\lfloor x \rfloor$ for the largest integer not exceeding $x$. All logarithms are taken to the  base $2$. $\mathbb E[\cdot ]$ denotes the expectation operator.  $H(\cdot)$ stands for entropy and $h(\cdot)$  for differential entropy.

\section{System model}
We consider a constant single-antenna $K$-user IC with channel matrix $\mathbf H =(h_{ij})_{1\leqslant i,j \leqslant K}\in\mathbb R^{K\times K}$, and input-output relation \begin{align}	Y_i={\sqrt{\snr}}\sum_{j=1}^K h_{ij} X_j +  Z_i,	\quad i=1, ... ,K ,	\label{eq:channel} \end{align}
where $X_i\in\mathbb R$ is the input at the $i$-th transmitter, $Y_i\in \mathbb R$ is the output at the $i$-th receiver, and $Z_i\in \mathbb R$ is  noise of absolutely continuous distribution with $h(Z_i)>-\infty$ and $H(\lfloor Z_i \rfloor)<\infty$. The input signals at different transmitters are independent and  noise is i.i.d.\ across  users and channel uses.

The channel matrix $\mathbf H$ is assumed to be  known perfectly at all transmitters and receivers. We impose the average power constraint
\begin{align*}	\frac{1}{n}\sum_{k=1}^n \left (x_{i}^{(k)}\right )^2\leqslant 1	
\end{align*}
on codewords $\left (x_{i}^{(1)} \,  ... \;\,  x_{i}^{(n)}\right )$ of block-length $n$ transmitted by user $i=1,...,K$.
The DoF  are defined as
\begin{align}	\Dof (\mathbf H) := \limsup_{\snr\to\infty}\frac{\overline C(\mathbf H; \snr)}{\frac{1}{2}\log \snr},	\label{eq:defdof1}	\end{align}
where $\overline C(\mathbf H; \snr)$ is the sum-capacity of the channel \eqref{eq:channel}.

\section{Main result}

We denote the vector consisting of the off-diagonal coefficients of $\mathbf H$ by $\mathbf  {\check h}\in \mathbb R^{K(K-1)}$, and  let $f_1,f_2, ...$ be the monomials of all degrees\footnote{The ``degree'' of a monomial is to be understood as the sum of all exponents of the variables involved (sometimes called the total degree).} in $K(K-1)$ variables enumerated as follows:   $f_1,...,f_{\varphi(d)}$  are  the monomials of degree  not larger than $d$, where  
\ba {\varphi(d)}:=\binom{K(K-1)+d}{d} .\ea


The main result of this paper is the following theorem.

\begin{theorem}\label{thm:explicit}
Suppose that the channel matrix $\mathbf H$ satisfies the following condition:
\begin{center}	For each $i=1,...,K$, the set \ba \{f_j(	\mathbf  {\check h}):j\geqslant 1  \}\cup \{h_{ii} f_j(	\mathbf  {\check h}):j\geqslant 1  \}\tag{$*$}\ea is linearly independent over $\mathbb Q$.	\end{center}
Then, we have \ba \Dof(\mathbf H) = K/2 . \ea
\end{theorem}
\begin{proof} 	See Section~\ref{sec:proof}.	\end{proof}
The proof of Theorem~\ref{thm:explicit} is constructive in the sense of providing an explicit construction of a sequence of input distributions that asymptotically achieves $\Dof(\mathbf H) =K/2$ for all $\mathbf H$ satisfying Condition ($*$).

Note  that the prominent example from \cite{EO09} with all entries of $\mathbf H$ rational, shown in \cite{EO09} to admit strictly less than $K/2$ DoF, does not satisfy Condition ($*$), since two rationals are always linearly dependent over $\mathbb Q$.


To see that Condition ($*$) is satisfied for (Lebesgue) almost all channel matrices, we first note that for  fixed $d\in \mathbb N$, fixed $a_1,...,a_{\varphi(d)}, b_1, ... , b_{\varphi(d)}\in \mathbb Z$, and fixed $i\in \{1,...,K\}$, 
\ban 	\sum_{j=1}^{\varphi(d)} a_jf_j(\mathbf  {\check h}) +\sum_{j=1}^{\varphi(d)} b_jh_{ii}f_j(\mathbf  {\check h}) =0 \label{eq:countable}\ean
is satisfied only on a set of measure $0$ with respect to the choice of $\mathbf H$. It suffices to consider linear combinations with coefficients in $\mathbb Z$ as \eqref{eq:countable} with rational coefficients can be multiplied by a  common denominator.
Since the set of equations \eqref{eq:countable} is countable with respect to  $d\in \mathbb N$, $a_1,...,a_{\varphi(d)}, b_1, ... , b_{\varphi(d)}\in \mathbb Z$, and $i\in\{1,...,K\}$, it follows that  Condition ($*$) is satisfied for almost all channel matrices $\mathbf H$.
Theorem~\ref{thm:explicit} therefore provides (Lebesgue) almost surely satisfied explicit conditions for $\mathbf H$ to admit $K/2$ DoF.

We proceed by developing, in Sections~\ref{sec:infdof} and \ref{sec:ifsinf}, preparatory material needed for the proof of Theorem~\ref{thm:explicit}.


\section{Information dimension and DoF}
\label{sec:infdof}

\begin{definition}\label{def:infdim} Let $X$ be a random variable with distribution $\mu$. We define the \emph{lower} and \emph{upper information dimension} of $X$ as
\begin{align*} \ldim(X) := \liminf_{k\to\infty}\frac{H(\lk X \rk _k)}{\log k} \;\: \text{and} \;\: \udim(X):= \limsup_{k\to\infty}\frac{H(\lk X \rk _k)}{\log k},  \end{align*}
where $\lk 
X \rk_k := \lfloor kX\rfloor /k$. If  $\ldim(X) = \udim(X)$, then we set $d(X):= \underline d(X) = \overline d(X)$ and  call $d(X)$ \emph{the information dimension} of $X$. Since $\ldim(X), \udim(X),$ and $d(X)$ depend on $\mu$ only, we sometimes also write $\ldim(\mu), \udim(\mu),$ and $d(\mu)$, respectively. 
\end{definition}

The relevance of information dimension in characterizing DoF stems from the following relation  \cite{SB12Allerton,GS07, WSV13}
\ban \limsup_{\snr\to\infty}\frac{h(\sqrt{\snr} X+Z)}{\frac{1}{2}\log\snr}=\udim(X), \label{eq:guio} \ean
which holds for arbitrary independent random variables $X$ and $Z$, with the distribution of $Z$  absolutely continuous and such that  $h(Z)>-\infty$, $H(\lfloor Z\rfloor )<\infty$. 

We can apply \eqref{eq:guio} to ICs as follows. By standard random coding arguments it follows for the IC \eqref{eq:channel}  that the sum-rate 
\ban I(X_1;Y_1) + \ldots +I(X_K;Y_K) \label{eq:sumrate} \ean
is achievable, where $X_1,..., X_K$ are independent input distributions with $\mathbb E[ X_i^2] \leqslant 1$, $i=1,...,K$. Using the chain rule, we get
\ban &I(X_i;Y_i) \label{eq:5} = h\lefto (Y_i \right ) - h\lefto (Y_i  \; \!\vert \;\! X_i\right ) \\ & \!\!\!\! = \! h\Bigg ({\sqrt{\snr}}\sum_{j=1}^K h_{ij} X_j +  Z_i \Bigg)\! - \! h\Bigg ({\sqrt{\snr}}\sum_{j\neq i}^K h_{ij} X_j +  Z_i \Bigg ) \label{eq:7}\ean
for $i=1,...,K$. Combining \eqref{eq:guio}-\eqref{eq:7}, we obtain
\begin{align} 	\dof(X_1, ... ,X_K ; \mathbf H) &:=  \nonumber \\&\!\!\!\!\! \sum_{i=1}^K \left [d\Bigg (\sum_{j=1}^K h_{ij} X_j \Bigg )-	 d \Bigg( \sum_{j\neq i}^K  h_{ij} X_j \Bigg)\right ]\label{eq:defdof} \\&\leqslant \Dof(\mathbf H), \label{eq:ineq1}   	 \end{align}
for all independent $X_1,...,X_K$ with\footnote{\label{fn:1}We only need the conditions $\mathbb E[ X_i^2] <\infty$ as scaling of the inputs does not affect $\dof(X_1, ... ,X_K ; \mathbf H) $.}  $\mathbb E[X_i^2]<\infty$, $i=1,...,K$, and such that all  information dimension terms appearing in \eqref{eq:defdof} exist. Relation \eqref{eq:ineq1} was first reported in \cite{WSV13}. 
A striking result in \cite{WSV13} shows that input distributions of discrete, continuous, or mixed discrete-continuous nature can achieve no more than $1$ DoF. For $K>2$, DoF-optimal input distributions therefore necessarily have a singular component. 


\section{Iterated function systems}
\label{sec:ifsinf}

A  class of singular distributions  with explicit expressions for their information dimension are self-similar distributions defined as follows.
Consider a finite set $\Phi_r:=\{\varphi_{i,r} : i=1,..., n\}$ consisting of  affine contractions of $\mathbb R$, i.e., \ban \varphi_{i,r}(x)=rx+w_i,\label{eq:ifs}\ean where $r\in I\subseteq (0,1)$ and the $w_i$ are pairwise different real numbers. We furthermore let $\mathcal W:=\{w_1,... ,w_n\}$. $\Phi_r$ is called an iterated function system (IFS)  parametrized by the contraction parameter $r\in I$. 
By classical fractal geometry \cite[Ch.~9]{Fal04} every IFS has an associated unique attractor, i.e., a non-empty compact set  $\mathcal A\subseteq \mathbb R$ such that \ba \mathcal A=\bigcup_{i=1}^n \varphi_{i,r}(\mathcal A) .\ea 
Moreover, for each probability vector $(p_1,...,p_n)$, there is a unique (Borel) probability distribution $\mu_r$ on $\mathbb R$ 
such that 
\ban \mu_r= \sum_{i=1}^n p_i (\varphi_{i,r})_\ast\mu_r , \label{eq:dist}\ean
where $(\varphi_{i,r})_\ast\mu_r$ denotes the push-forward of $\mu_r$ by $\varphi_{i,r}$. The distribution $\mu_r$ is supported on $\mathcal A$ and is called the self-similar distribution corresponding to the IFS $\Phi_r$ with underlying probability vector $(p_1,...,p_n)$. We can give the following explicit expression for a random variable $X$ with distribution $\mu_r$ in \eqref{eq:dist}
\ban X = \sum_{k=0}^\infty r^k W_k,		\label{eq:representation} \ean
where $\{W_k\}_{k\geqslant 0}$ is a set of  i.i.d.\ copies of a random variable $W$ drawn from the set $\mathcal W$ according to  $(p_1,...,p_n)$.

\section{The main ideas}

Classical results in  fractal geometry allow an analytical expression for the information dimension of a self-similar distribution under the so-called open set condition  \cite[Thm.~2]{GH89}. This condition requires the existence of a non-empty open set whose images under the elements of $\Phi_r$ do not overlap and all lie in the open set  (see, e.g., \cite{BHR05}). Wu et al.\ \cite{WSV13} ensure that the open set condition is satisfied by imposing an upper bound on the contraction parameter $r$ according to 
\ban r\leqslant \frac{\mathsf m (\mathcal W)}{\mathsf m (\mathcal W)+\mathsf M (\mathcal W)}. \label{eq:suff} \ean
Here, $\mathsf m (\mathcal W):=\min_{i\neq j} |w_i-w_j|$ and $\mathsf M (\mathcal W):=\max_{i,j} |w_i-w_j|$. 
The authors of \cite{WSV13} construct $K/2$ DoF-achieving input distributions by building $\mathcal W$
 from $\mathbb Z$-linear combinations of monomial expressions in the off-diagonal channel coefficients; this idea inspired our Condition ($*$).  However, to ensure that the construction in \cite{WSV13}  meets \eqref{eq:suff}  the minimum and maximum distance of the elements in $\mathcal W$ needs to be controlled. This results in the question of how fast a polynomial in real variables with integer coefficients approaches an integer \cite{WSV13}, a problem  studied in Diophantine approximation theory. The nature of the results applied in \cite{WSV13} to deal with this question does not seem to allow  an explicit characterization of the channel matrices that admit $K/2$ DoF.
Recent groundbreaking work by Hochman \cite{Hoc12} replaces the open set condition by a much weaker condition, which essentially requires that the IFS must not  allow ``exact overlap''. This improvement turns out to be instrumental in the proof of Theorem~\ref{thm:explicit}. 
Specifically, we use the following simple consequence of a key result by Hochman \cite[Thm.~1.8]{Hoc12}.


\begin{theorem} \label{thm:hochman}
If $I\subseteq (0,1)$ is a non-empty compact interval, and $\mu_r$ is the self-similar distribution with underlying contraction parameter $r\in I$  and  probability vector $(p_1,...,p_n)$, then\footnote{The ``$1$'' in the minimum simply accounts for the fact that information dimension cannot exceed the dimension of the ambient space.} 
\ban d(\mu_r)= \min \lefto \{\frac{\sum p_i \log p_i}{\log r}, 1\right \} , \label{eq:formula} \ean
for all $r\in I\! \setminus\! E$, where $E$ is a set of Hausdorff and packing dimension $0$.
\end{theorem}
\begin{proof}
For $\mathbf i\in \{1,...,n\}^k$, let $\varphi_{\mathbf i,r}:=\varphi_{i_1,r}\circ \ldots \circ \varphi_{i_k,r}$ and define 
\ba \Delta_{\mathbf i,\mathbf j}(r) := \varphi_{\mathbf i, r}(0) - \varphi_{\mathbf j, r}(0),  \ea
for  $\mathbf i,\mathbf j\in \{1,...,n\}^k$. Extend this definition to infinite sequences $\mathbf i,\mathbf j\in \{1,...,n\}^\mathbb N$ according to
\ba \Delta_{\mathbf i,\mathbf j}(r) := \lim_{k\to\infty}\Delta_{(i_1,...,i_k),(j_1,...,j_k)}(r).  \ea
Using  \eqref{eq:ifs} it follows that 
\ba \Delta_{\mathbf i,\mathbf j}(r) =\sum_{k=1}^\infty r^{k-1} (w_{i_k}-w_{j_k}). \ea
 Since the $w_i$ are pairwise different and a power series can  vanish on a non-empty interval only if it is identically zero, we get that $\Delta_{\mathbf i,\mathbf j}\equiv 0$ on a compact interval $I$ if and only if $\mathbf i=\mathbf j$. The result now follows by application of \cite[Thm.~1.8]{Hoc12}.
\end{proof}

\begin{remark}
Note that   \eqref{eq:formula} can be rewritten in terms of the entropy of the random variable $W$, defined implicitly in \eqref{eq:representation}, which takes value $w_i$ with probability $p_i$:
\ban d(\mu_r)= \min \lefto \{\frac{H(W)}{\log (1/r)}, 1\right \} .\label{eq:formula2}\ean
\end{remark}


We wish to construct self-similar input distributions that yield $\dof(X_1,...,X_K;\mathbf H)=K/2$ for all channel matrices satisfying Condition~($*$). To this end, we first note that choosing all inputs $X_i$ to have self-similar distributions as in \eqref{eq:representation} with identical contraction parameter $r$, the distributions of the random variables $\sum  h_{ij} X_j $ appearing in \eqref{eq:defdof} are again self-similar. This  allows us to compute the information dimension terms in \eqref{eq:defdof} using \eqref{eq:formula2}. The  freedom we exploit in constructing full DoF-achieving $X_i$ lies in the choice of the set $\mathcal W$, $(p_1,...,p_n)$ is assumed uniform on $\mathcal W$, for simplicity of exposition. Specifically, we want to ensure that  the first term inside the sum \eqref{eq:defdof} is twice as big as the second term, for all $i$. This means that the sum of the desired signal $h_{ii}X_i$ and the interference $\sum_{j\neq i}  h_{ij} X_j$ should be  twice as ``rich'' as the interference term alone. 
From \eqref{eq:formula2} it follows that this doubling is accomplished if
\ban 
\Bigg | \sum_{j=1}^Kh_{ij}\mathcal W \Bigg | &\approx \Bigg |\sum_{j\neq i}^Kh_{ij}\mathcal W \Bigg | ^2. \label{eq:key}\ean
Let us now turn to the main idea for realizing \eqref{eq:key}.  We build $\mathcal W$ as a set of $\mathbb Z$-linear combinations of monomial expressions in the off-diagonal channel coefficients, an idea that was introduced in \cite{MGMK09}. 
Multiplying the elements of the so-constructed set $\mathcal W$ by an off-diagonal channel coefficient simply increases the degrees of the involved monomials by $1$ so that the algebraic structure of $h_{ij}\mathcal W$, for $i\neq j$, is the same as that of $\mathcal W$. In addition,  we let the degrees of the involved monomials be large and say that multiplication of $\mathcal W$ by $h_{ij}$, $i\neq j$, results in a set that is ``similar'' to $\mathcal W$, denoted as $h_{ij}\mathcal W\approx \mathcal  W$.
Finally, we also take $|\mathcal W|$ to be large so that $\sum_{j\neq i}h_{ij}\mathcal W\approx \mathcal W$,   see Fig.~\ref{fig:bla}~(a).


As $h_{ii}$ does not participate in the monomial expressions in $\mathcal W$, the set $h_{ii}\mathcal W$ is unlikely to have the same algebraic structure as $\mathcal W$. In fact, Condition ($*$) guarantees that this does not happen. We can then  conclude that the set $ h_{ii}\mathcal W+ \sum_{j\neq i}h_{ij}\mathcal W \approx h_{ii}\mathcal W +\mathcal W$ roughly has cardinality $\left |\mathcal W \right |^2$, as desired, see Fig.~\ref{fig:bla}~(b). Moreover, Condition ($*$) guarantees that every  element of $\mathcal W$ has exactly one representation as a $\mathbb Z$-linear combination of the monomials in the off-diagonal channel coefficients. This lets us control the cardinality of $\mathcal W$ by controlling the number of  possible $\mathbb Z$-linear combinations.

\begin{figure}
\begin{subfigure}{.5\textwidth}
\begin{center}
\begin{tikzpicture}[scale=.7]
\begin{scope}
\draw (0,0) circle (2pt);
\node at (-.18,-.3) {$0$};
\draw (1,0) circle (2pt);
\draw (.5,.5) circle (2pt);
\draw (.5,-.5) circle (2pt);
\draw (-.5,.5) circle (2pt);
\draw (-.5,-.5) circle (2pt);
\draw (-1,0) circle (2pt);
\end{scope}
\node at (2,0) {$+$};
\begin{scope}[xshift=4cm]
\draw (0,0) circle (2pt);
\node at (-.18,-.3) {$0$};
\draw (1,0) circle (2pt);
\draw (.5,.5) circle (2pt);
\draw (.5,-.5) circle (2pt);
\draw (-.5,.5) circle (2pt);
\draw (-.5,-.5) circle (2pt);
\draw (-1,0) circle (2pt);
\end{scope}
\node at (6,0) {$=$};
\begin{scope}[xshift=9cm]
\draw (0,0) circle (2pt);
\node at (-.18,-.3) {$0$};
\draw (1,0) circle (2pt);
\draw (.5,.5) circle (2pt);
\draw (.5,-.5) circle (2pt);
\draw (-.5,.5) circle (2pt);
\draw (-.5,-.5) circle (2pt);
\draw (-1,0) circle (2pt);
\draw (2,0) circle (2pt);
\draw (1.5,-.5) circle (2pt);
\draw (1,-1) circle (2pt);
\draw (0,-1) circle (2pt);
\draw (-1,-1) circle (2pt);
\draw (-1.5,-.5) circle (2pt);
\draw (-2,0) circle (2pt);
\draw (-1.5,.5) circle (2pt);
\draw (-1,1) circle (2pt);
\draw (0,1) circle (2pt);
\draw (1,1) circle (2pt);
\draw (1.5,.5) circle (2pt);
\end{scope}
\end{tikzpicture}
\end{center}

\caption{Sum of two sets with common algebraic structure.}
\label{fig:bla1}

\end{subfigure}
\vspace{.2cm}

\begin{subfigure}{.5\textwidth}
\begin{center}
\begin{tikzpicture}[scale=.7]
\begin{scope}
\draw (0,0) circle (2pt);
\node at (-.18,-.3) {$0$};
\draw (.8,0) circle (2pt);
\draw (.4,.4) circle (2pt);
\draw (.4,-.4) circle (2pt);
\draw (-.4,.4) circle (2pt);
\draw (-.4,-.4) circle (2pt);
\draw (-.8,0) circle (2pt);
\end{scope}
\node at (2,0) {$+$};
\begin{scope}[xshift=4cm]
\draw (0,0) circle (2pt);
\node at (-.18,-.3) {$0$};
\draw (1,0) circle (2pt);
\draw (.5,.5) circle (2pt);
\draw (.5,-.5) circle (2pt);
\draw (-.5,.5) circle (2pt);
\draw (-.5,-.5) circle (2pt);
\draw (-1,0) circle (2pt);
\end{scope}
\node at (6,0) {$=$};
\begin{scope}[xshift=9cm]
\draw (0,0) circle (2pt);
\node at (-.18,-.3) {$0$};
\draw (1,0) circle (2pt);
\draw (.5,.5) circle (2pt);
\draw (.5,-.5) circle (2pt);
\draw (-.5,.5) circle (2pt);
\draw (-.5,-.5) circle (2pt);
\draw (-1,0) circle (2pt);

\draw (.8,0) circle (2pt);
\draw (1.8,0) circle (2pt);
\draw (1.3,.5) circle (2pt);
\draw (1.3,-.5) circle (2pt);
\draw (.3,.5) circle (2pt);
\draw (.3,-.5) circle (2pt);
\draw (-.2,0) circle (2pt);

\draw (-.8,0) circle (2pt);
\draw (.2,0) circle (2pt);
\draw (-.3,.5) circle (2pt);
\draw (-.3,-.5) circle (2pt);
\draw (-1.3,.5) circle (2pt);
\draw (-1.3,-.5) circle (2pt);
\draw (-1.8,0) circle (2pt);

\draw (.4,.4) circle (2pt);
\draw (1.4,.4) circle (2pt);
\draw (.9,.9) circle (2pt);
\draw (.9,-.1) circle (2pt);
\draw (-.1,.9) circle (2pt);
\draw (-.1,-.1) circle (2pt);
\draw (-.6,0.4) circle (2pt);

\draw (.4,-.4) circle (2pt);
\draw (1.4,-.4) circle (2pt);
\draw (.9,.1) circle (2pt);
\draw (.9,-.9) circle (2pt);
\draw (-.1,.1) circle (2pt);
\draw (-.1,-.9) circle (2pt);
\draw (-.6,-.4) circle (2pt);

\draw (-.4,.4) circle (2pt);
\draw (.6,.4) circle (2pt);
\draw (.1,.9) circle (2pt);
\draw (.1,-.1) circle (2pt);
\draw (-.9,.9) circle (2pt);
\draw (-.9,-.1) circle (2pt);
\draw (-1.4,.4) circle (2pt);

\draw (-.4,-.4) circle (2pt);
\draw (.6,-.4) circle (2pt);
\draw (.1,.1) circle (2pt);
\draw (.1,-.9) circle (2pt);
\draw (-.9,.1) circle (2pt);
\draw (-.9,-.9) circle (2pt);
\draw (-1.4,-.4) circle (2pt);

\end{scope}
\end{tikzpicture}
\end{center}

\caption{Sum of two sets with different algebraic structures.}
\label{fig:bla2}
\end{subfigure}
\caption{\small The cardinality of the sum in  (a) is $19$ and hence  small compared to the $7^2=49$  pairs  summed up, whereas the sum in  (b) has cardinality $49$.}
\label{fig:bla}
\end{figure}
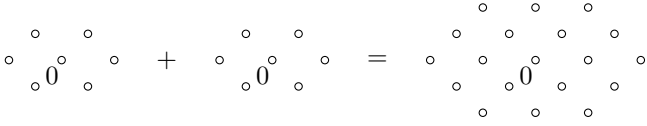
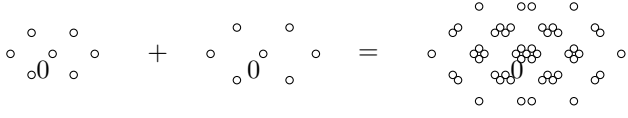

\addtolength{\textheight}{-1cm}

\section{Proof of Theorem~\ref{thm:explicit} }
\label{sec:proof}

We first note that Condition ($*$)  implies that $\mathbf H$ cannot contain zeros, i.e., $\mathbf H$ must be fully connected.
The upper bound $\Dof(\mathbf H)\leqslant K/2$ therefore follows from  \cite[Prop.~1]{HMN05}.

The remainder of the proof establishes the lower bound  $\Dof(\mathbf H)\geqslant K/2$.
Let $N$ and $d$ be positive integers. We begin by setting
\ban \mathcal W_N:=  \Bigg \{ \sum_{i=1}^{\varphi(d)} a_i f_i(\mathbf  {\check h}) \; : \; a_1, ..., a_{\varphi(d)}\in\{1,..., N\}  \Bigg \} \label{eq:in} \ean
and $r_N:=|\mathcal W_N|^{-2}$. Let $\{W_{i,k} \; : \; 1\leqslant i\leqslant K, k\geqslant 0\}$ be i.i.d.\ uniform random variables on $\mathcal W_N$. For $\varepsilon>0$, we set $I_\varepsilon:=[r_N-\varepsilon, r_N]$ and consider the inputs 
\ba X_i=\sum_{k=0}^\infty {r}^kW_{i,k}, \quad i=1,...,K,\ea
where $r\in I_\varepsilon$. Then, the signals 
\ba 		\sum_{j=1}^Kh_{ij}X_j &=\sum_{k=0}^\infty r^k \sum_{j=1}^Kh_{ij}W_{j,k} \\ \text{and} \quad \sum_{j\neq i}^Kh_{ij}X_j&=\sum_{k=0}^\infty r^k \sum_{j\neq i}^Kh_{ij}W_{j,k}\ea
also have self-similar distributions with alphabets $ \sum_{j=1}^Kh_{ij}\mathcal W_N $ and $\sum_{j\neq i}h_{ij}\mathcal W_N$, respectively.
Applying Theorem~\ref{thm:hochman}, we can therefore conclude the existence of an  $\tilde r\in I_\varepsilon$ such that for $X_i=\sum_{k=0}^\infty {\tilde r}^kW_{i,k}$,  $i=1,...,K$, we have
\ba 	d\Bigg (\sum_{j=1}^Kh_{ij}X_j\Bigg )&= \min \lefto \{\frac{H\lefto (\sum_{j=1}^Kh_{ij}W_{j,0}\right )}{\log (1/ \tilde r)}, 1\right \}	\\  \text{and} \quad d\Bigg (\sum_{j\neq i}^Kh_{ij}X_j\Bigg)&=\min \lefto \{\frac{H\lefto (\sum_{j\neq i}^K h_{ij}W_{j,0}\right )}{\log (1/ \tilde r)}, 1\right \}	,	\ea
for $i=1,...,K$. Taking $\varepsilon \to 0$, we get $\tilde r \to r_N$, and  by the continuity of $\log (\cdot )$ and  \eqref{eq:ineq1}   it follows that 
\ban \sum_{i=1}^K  &\left [  \min \lefto \{\frac{H\lefto (\sum_{j=1}^Kh_{ij}W_{j,0}\right )}{\log (1/  r_N)}, 1\right \} \right .   \nonumber  \\ & \left .- \min \lefto \{\frac{H\lefto (\sum_{j\neq i}^K h_{ij}W_{j,0}\right )}{\log (1/  r_N)}, 1\right \}\right ] \leqslant \Dof(\mathbf H).\label{eq:3} \ean
Note that the random variable $\sum_{j\neq i} h_{ij}W_{j,0}$ takes value in 
\ban \Bigg\{\! \sum_{i=1}^{\varphi(d+1)}a_i f_i(\mathbf  {\check h}) \; : \; a_1, ..., a_{\varphi(d+1)}\in\{1,..., (K-1)N\}  \Bigg \} . \label{eq:pres} \ean
By Condition ($*$) it follows that each element in the set \eqref{eq:pres} has exactly one representation as a $\mathbb Z$-linear combination with coefficients $a_1, ..., a_{\varphi(d+1)}\in\{1,..., (K-1)N\} $. This allows us to conclude that the cardinality of the set \eqref{eq:pres} is given by $((K-1)N)^{\varphi(d+1)}$, which implies $H\lefto (\sum_{j\neq i} h_{ij}W_{j,0}\right )\leqslant {\varphi(d+1)} \log((K-1)N) $.
With $\log (1/r_N)=2 \log |\mathcal W_N|=2 {\varphi(d)} \log N  $, we therefore get
\ban \frac{H\lefto (\sum_{j\neq i}^K h_{ij}W_{j,0}\right )}{\log (1/  r_N)} \leqslant \frac{{\varphi(d+1)}\log((K-1)N) }{2{\varphi(d)} \log N} \xrightarrow{d,N\to \infty} \frac{1}{2}, \label{eq:2} \ean 
where we used \ba \frac{\varphi(d+1)}{\varphi(d)}=\frac{K(K-1)+d+1}{d+1}\xrightarrow{d\to \infty}1.\ea
Next, note that Condition ($*$) implies that the sum $h_{ii} W_{i,0} + \sum_{j\neq i}h_{ij}W_{j,0}$ can be  separated uniquely into the terms $h_{ii} W_{i,0} $ and $\sum_{j\neq i}h_{ij}W_{j,0}$ in the sense that the pair $( h_{ii} W_{i,0} , \sum_{j\neq i}h_{ij}W_{j,0})$ and the sum $h_{ii} W_{i,0} + \sum_{j\neq i}h_{ij}W_{j,0}$ are related through a bijection.
It therefore follows   that the pair and the sum have equal entropy \cite[Ex.~2.4]{CT06}, and since the $W_{j,0}$, $1\leqslant j \leqslant K$, are independent, we find that
\ban \! H\Bigg (\sum_{j=1}^Kh_{ij}W_{j,0}\Bigg )&= H \Bigg (h_{ii} W_{i,0}, \sum_{j\neq i}h_{ij}W_{j,0}	\Bigg )\\ &= H \lefto (h_{ii} W_{i,0} \right ) + H \Bigg ( \sum_{j\neq i}h_{ij}W_{j,0}	\label{eq:decomp}\Bigg ).\ean 
%
Putting the pieces together, we obtain
\ban  &\frac{H\lefto (\sum_{j=1}^Kh_{ij}W_{j,0}\right ) - H\lefto (\sum_{j\neq i}^K h_{ij}W_{j,0}\right ) }{\log (1/  r_N)}\\ &= \frac{H(h_{ii}W_{i,0})}{2{\varphi(d)}\log N}= \frac{{\varphi(d)}\log N}{2{\varphi(d)}\log N} =\frac{1}{2},\label{eq:1} \ean 
where we used the scaling invariance of entropy, the fact that $W_{i,0}$ is uniform on $\mathcal W$, and $|\mathcal W|=N^{\varphi(d)}$. Finally, \eqref{eq:1} and \eqref{eq:2} imply that the left-hand side of \eqref{eq:3} tends to $\sum_{i=1}^K(1-\frac{1}{2})=K/2$  for $d,N\to \infty$. This completes the proof.
\endproof



\section{Condition ($*$) is not necessary}

While Condition ($*$) is  sufficient  for $\Dof(\mathbf H)=K/2$, we next show that it is not necessary. This will be accomplished by constructing a class of example channels  which fail to satisfy Condition ($*$) but still admit $K/2$ DoF.  As almost every channel matrix  satisfies Condition ($*$) this example class necessarily has Lebesgue measure zero.
Specifically, we consider channel matrices that have $h_{ii} \in \mathbb R\! \setminus\! \mathbb Q$, $i=1,...,K$, and $h_{ij}\in \mathbb Q\! \setminus\! \{0\}$, for $i,j=1,...,K$ with  $i\neq j$. Since this means, in particular,  that all channel coefficients are nonzero, $\mathbf H$ is fully connected. The upper bound $\Dof(\mathbf H)\leqslant K/2$ is therefore again implied by  \cite[Prop.~1]{HMN05}. Moreover, as two rational numbers are linearly dependent over $\mathbb Q$, these channel matrices do not satisfy Condition ($*$). 
We next show that this example class nevertheless has $\Dof(\mathbf H)\geqslant K/2$ and hence $\Dof(\mathbf H)= K/2$. This will be accomplished by constructing corresponding asymptotically DoF-optimal input distributions. First, we argue that we may assume $h_{ij}\in\mathbb Z$, for $i\neq j$. 
This follows from the fact that multiplying the channel coefficients $h_{ij}$, $i,j=1,...,K$, by a common denominator of the $h_{ij}$, $i\neq j$, is equivalent to scaling of the inputs $X_i$, which does not impact DoF.

 Let 
\ba \mathcal W:= \{ 0,..., N-1\} ,\ea 
for some $N>0$, and let $\{W_{i,k} \; : \; 1\leqslant i\leqslant K, k\geqslant 0\}$  be i.i.d.\ uniform on $\mathcal W$. We set the contraction parameter to
\ban r= 2^{-2\log(2h_\text{max} KN)}, \label{eq:r} \ean where $h_\text{max}:=\max \{ |h_{ij}| \, : \, i\neq j\}$. Since $\sum_{j\neq i} h_{ij}W_{j,0}$ is integer-valued and $h_{ii}W_{i,0}$ is irrational, the sum $\sum_{j=1}^K h_{ij}W_{j,0}$ can  uniquely be separated into the terms $h_{ii}W_{i,0}$ and $\sum_{j\neq i} h_{ij}W_{j,0}$ in the sense of  \eqref{eq:decomp}. 
As $h_{ij}\in\mathbb Z$ for $i\neq j$, we have 
\ba \sum_{j\neq i}^K h_{ij}W_{j,0} \in \{ -h_\text{max} (K-1) N, ... ,0 , ... ,h_\text{max}(K-1)N\} \ea
and hence $ H\lefto (\sum_{j\neq i} h_{ij}W_{j,0}\right ) \leqslant \log \lefto (2h_\text{max}KN \right )$. 
Since the $W_{j,0}$, $ 1\leqslant j \leqslant K$, are i.i.d., we furthermore have $H(h_{ii}W_{i,0})\leqslant H (\sum_{j\neq i} h_{ij}W_{j,0} )$  \cite[Ex.~2.14]{CT06} and  \eqref{eq:decomp}  implies that
\ba H \Bigg (\sum_{j=1}^K h_{ij}W_{j,0}\Bigg )
\leqslant 2 H\Bigg(\sum_{j\neq i}^K h_{ij}W_{j,0}\Bigg )\leqslant  2\log \lefto (2h_\text{max}KN \right ). \ea
With \eqref{eq:r} we therefore obtain \ba \min \lefto \{ \frac{H\lefto (\sum_{j=1}^K h_{ij}W_{j,0}\right )}{\log (1/r)}, 1\right \}= \frac{H\lefto (\sum_{j=1}^K h_{ij}W_{j,0}\right )}{\log (1/r)}. \ea Application of Theorem~\ref{thm:hochman} now yields that for each $\varepsilon>0$, there is an $\tilde r\in I_\varepsilon:=[r-\varepsilon, r]$ such that  the inputs $X_i=\sum_{k=0}^\infty {\tilde r}^kW_{i,k}$,  for $i=1,...,K$, satisfy
\ba 	d\Bigg (\sum_{j=1}^Kh_{ij}X_j\Bigg )\! - d\Bigg (\sum_{j\neq i}^Kh_{ij}X_j\Bigg )&\! =\! \frac{H(h_{ii}W_{i,0})}{\log (1/\tilde r)}\! =\!\frac{\log N}{\log (1/\tilde r)}	,	\ea
where we used \eqref{eq:decomp}. As $\varepsilon$ can be made arbitrarily small, we find that
\ban 	\Dof(\mathbf H)\geqslant \frac{K\log N}{\log (1/r)}=\frac{K\log N}{2\log \lefto (2h_\text{max}KN \right )} .	\label{eq:lhs} \ean
Letting $N\to \infty$, the right-hand side of \eqref{eq:lhs} approaches $K/2$, which completes the proof.

We conclude by noting that the example class studied here was investigated before in \cite[Thm.~1]{EO09} and \cite[Thm.~6]{WSV13}. In contrast to \cite{EO09, WSV13} our proof of DoF-optimality  is, however, not based on arguments from Diophantine approximation theory. 

\bibliographystyle{IEEEtran}
\bibliography{IEEEabrv,refs}

\begin{thebibliography}{10}
\providecommand{\url}[1]{#1}
\csname url@samestyle\endcsname
\providecommand{\newblock}{\relax}
\providecommand{\bibinfo}[2]{#2}
\providecommand{\BIBentrySTDinterwordspacing}{\spaceskip=0pt\relax}
\providecommand{\BIBentryALTinterwordstretchfactor}{4}
\providecommand{\BIBentryALTinterwordspacing}{\spaceskip=\fontdimen2\font plus
\BIBentryALTinterwordstretchfactor\fontdimen3\font minus
  \fontdimen4\font\relax}
\providecommand{\BIBforeignlanguage}[2]{{%
\expandafter\ifx\csname l@#1\endcsname\relax
\typeout{** WARNING: IEEEtran.bst: No hyphenation pattern has been}%
\typeout{** loaded for the language `#1'. Using the pattern for}%
\typeout{** the default language instead.}%
\else
\language=\csname l@#1\endcsname
\fi
#2}}
\providecommand{\BIBdecl}{\relax}
\BIBdecl

\bibitem{Hoc12}
M.~Hochman, ``On self-similar sets with overlaps and inverse theorems for
  entropy,'' \emph{Annals of Mathematics}, {to appear}.

\bibitem{Jaf11}
S.~A. Jafar, ``Interference alignment --- {A} new look at signal dimensions in
  a communication network,'' \emph{Foundations and Trends in Communications and
  Information Theory}, vol.~7, no.~1, 2011.

\bibitem{MGMK09}
A.~S. Motahari, S.~O. Gharan, M.-A. Maddah-Ali, and A.~K. Khandani, ``Real
  interference alignment: Exploiting the potential of single antenna systems,''
  \emph{arXiv:0908.2282 [cs.IT]}, {Aug.} 2009.

\bibitem{EO09}
R.~H. Etkin and E.~Ordentlich, ``The degrees-of-freedom of the {K}-user
  {Gaussian} interference channel is discontinuous at rational channel
  coefficients,'' \emph{IEEE Trans. Inf. Theory}, vol.~55, no.~11, pp.
  4932--4946, {Nov.} 2009.

\bibitem{WSV13}
Y.~Wu, S.~Shamai~(Shitz), and S.~Verd\'u, ``A formula for the degrees of
  freedom of the interference channel,'' \emph{Submitted to IEEE Trans. Inf.
  Theory}, {revised version, Jun.} 2013.

\bibitem{BHR05}
C.~Brandt, N.~Viet~Hung, and H.~Rao, ``On the open set condition for
  self-similar fractals,'' \emph{Proc. of the AMS}, vol. 134, no.~5, pp.
  1369--1374, {Oct.} 2005.

\bibitem{SB12Allerton}
D.~Stotz and H.~B\"olcskei, ``Degrees of freedom in vector interference
  channels,'' \emph{Proc. 50th Ann. Allerton Conf. on Communication, Control,
  and Computing}, pp. 1755--1760, {Oct.} 2012.

\bibitem{GS07}
A.~Guionnet and D.~Shlyakhtenko, ``On classical analogues of free entropy
  dimension,'' \emph{Journal of Functional Analysis}, vol. 251, pp. 738--771,
  {Oct.} 2007.

\bibitem{Fal04}
K.~Falconer, \emph{Fractal Geometry: Mathematical Foundations and
  Applications}, {2nd}~ed.\hskip 1em plus 0.5em minus 0.4em\relax John Wiley \&
  Sons, 2004.

\bibitem{GH89}
J.~S. Geronimo and D.~P. Hardin, ``An exact formula for the measure dimensions
  associated with a class of piecewise linear maps,'' \emph{Constructive
  Approximation}, vol.~5, pp. 89--98, {Dec.} 1989.

\bibitem{HMN05}
A.~H{\o}st-Madsen and A.~Nosratinia, ``The multiplexing gain of wireless
  networks,'' \emph{Proc. IEEE Int. Symp. on Inf. Theory}, pp. 2065--2069,
  {Sep.} 2005.

\bibitem{CT06}
T.~M. Cover and J.~A. Thomas, \emph{Elements of Information Theory},
  2nd~ed.\hskip 1em plus 0.5em minus 0.4em\relax New York, NY:
  Wiley-Interscience, 2006.

\end{thebibliography}
\end{document}